\documentclass[11pt]{amsart}
\usepackage{fullpage}
\usepackage{amsthm,amsmath,amssymb}
\usepackage[foot]{amsaddr}

\usepackage{enumitem}
\usepackage{booktabs}
\usepackage{graphicx}
\usepackage{subcaption}
\usepackage{tikz}
\usepackage{tikz-qtree}
 \usepackage{qtree}
\usepackage{gb4e}
\usetikzlibrary{fadings}
\tikzfading[name=fade down,top color=blue!0, bottom color=red!100]
\tikzfading[name=fade up,top color=red!100, bottom color=blue!0]

\usepackage[colorlinks=true,linkcolor=red,breaklinks=true,citecolor=blue,urlcolor=cyan]{hyperref}

\usepackage{caption} %% GR
\usepackage{algorithm} %% GR
\usepackage[noEnd]{algpseudocodex} %% GR

\graphicspath{{./figs/}} %note trailing /

\newtheorem{theorem}{Theorem}[section]
\newtheorem{proposition}[theorem]{Proposition}

\theoremstyle{remark}

\DeclareMathOperator{\argmax}{argmax}

\title{Branch and Bound to Assess Stability of Regression Coefficients in Uncertain Models}

\author{Brian Knaeble}
\address{Department of Computer Science, Utah Valley University, Orem, UT}
\author{R. Mitchell Hughes}
\email{bknaeble@uvu.edu}
\author{George Rudolph}
\author{Mark A. Abramson}
\address{Department of Mathematics, Utah Valley University, Orem, UT}
\email{mark.abramson@uvu.edu}
\author{Daniel Razo}

\begin{document}

\begin{abstract}
It can be difficult to interpret a coefficient of an uncertain model. A slope coefficient of a regression model may change as covariates are added or removed from the model. In the context of high-dimensional data, there are too many model extensions to check. However, as we show here, it is possible to efficiently search, with a branch and bound algorithm, for maximum and minimum values of that adjusted slope coefficient over a discrete space of regularized regression models. %Our introduced branch and bound algorithm allows us to search approximately one million times as efficiently as a brute force search. 
Here we introduce our algorithm, along with supporting mathematical results, an example application, and a link to our computer code, to help researchers summarize high-dimensional data and assess the stability of regression coefficients in uncertain models.
\end{abstract}

\maketitle

\section{Introduction}
\label{introduction}
The principle of least squares can be traced back to Gauss \cite{Gauss}. There is even a ``loose correspondence'' between polynomial regression and neural networks, and polynomial regression may be more interpretable \cite{Cheng}. %\textcolor{brown}{When performing least squares regression, we may interpret the *regression* coefficient as the effect of the explanatory variable on the response variable. However, the magnitude of the effect depends...}
However, outside the context of a randomized, controlled experiment, interpretation of a regression coefficient can be problematic. Within the context of an observational study there may be substantial model uncertainty when a regression model is fit to high dimensional data. A slope coefficient corresponding to a particular explanatory variable will depend on which subset of covariates has been selected to fit the model \cite{KD}. It may not be wise to include every covariate within the model \cite{DM}. If there are $p$ covariates under consideration then there are at least $2^p$ models to choose from.

Patel et al.\cite{Patel2015} have fit $2^{13}$ model extensions within the context of logistic regression, reporting wide variation of fitted coefficients, a phenomenon they describe as vibration of effects. Random selection of extensions is sufficient to demonstrate the existence of model uncertainty, but computing a standard error of estimates from random extensions would be insufficient for the purpose of quantifying uncertainty of an interpretation. The large number $2^{13}$ of randomly selected model extensions is relatively small compared to $2^p$ when $p$ is large, and there is no guarantee that the randomly selected model extensions inform accurate interpretation. In this paper we will require a model that is linear in its parameters, but we will be able to search over a much larger space of $2^{30}$ or more model extensions to find maximum and minimum effect estimates.  %in that case it may be unlikely to randomly select an admissible set of covariates \cite[Section 3.3.1]{Pearl2009}. Here we will require a model that is linear in its parameters, but we will be able to search over $2^{30}$ or more model extensions to find maximum and minimum effect estimates. 
%There are $2^p$ subsets of covariates to choose from, and even more models if we allow higher order terms, i.e. by conducting polynomial regression. Patel et al.\cite{Patel2015} fit $2^{13}$ model extensions within the context of logistic regression, and reported wide variation of fitted coefficients, a phenomenon they described as vibrations of effects. Here we require a model that is linear in its parameters, but we will be able to search over $2^{30}$ or more model extensions to find maximum and minimum effect estimates. %That is an approximately $2^{20}$ fold increase in efficiency over a brute force search.

There are assumptions that are often made to support causal interpretation. Within the context of this paper we could assume, conditional on a particular subset of measured covariates, that (natural) ``assignment'' of an explanatory variable is strongly ignorable, while also assuming the stable unit treatment value assumption (SUTVA) \cite[Chapter 3]{Imbens2015}. Alternatively, we could assume that a particular subset of the measured covariates constitutes an admissible set \cite[Section 3.3.1]{Pearl2009}. It is not true that conditioning on more covariates is always better for the purpose of causal inference \cite{Pim}. We may identify the causal effect with a slope coefficient if we assume that a {\em particular} subset of measured covariates is admissible, but that is a strong assumption. We may partially identify (see \cite{Tamer2010}) the causal effect with an interval that bounds slope coefficients of uncertain models without specifying that any particular subset is admissible but only by assuming the {\em existence} of an admissible subset of measured covariates and then utilizing the introduced methodology of this paper. %Our introduced methodology supports causal inference and also more general interpretation of coefficients in uncertain models (cite math interpretation, cite interpetable AI).
%
%While our methodology was designed with causal inference in mind, it's ap may be broadly applicable and support interpretable AI more broadly, and it could be understood within the context of interpretable AI in broader context,  %It may be reasonable to trade (apparent) precision for more certain identification, but our introduced algorithm may have other uses outside of causal inference.

Our main contribution is our introduced branch and bound algorithm that can be used to assess the stability of regression coefficients in uncertain models; see Section \ref{methods}. Our secondary contribution is a technical and general mathematical result; see Proposition \ref{Danny}. That mathematical result is interesting on its own, and it supports our algorithm. We think of our algorithm less as a methodology for causal inference and more as a technique to summarize a data set by computing bounds for effect estimates that are consistent with the data. % like the data set pictured in Table \ref{tab1}, 
The input to the algorithm is the data. The output of the algorithm is both the maximum and minimum effect estimates, computed over a discrete space of regularized models each fit with the principle of least squares.

Our paper is organized as follows. We introduce notation and describe our algorithm in Section \ref{methods}. A mathematical result and the results of trials are described in Section \ref{results}. % , where we also describe our algorithm and provide a link to our computer code that implements the algorithm in R. 
An example application is described in Section \ref{applications}. A discussion is provided in Section \ref{discussion}. In Appendix \ref{AppA} we describe the details of the data sets that we utilized during trials of our algorithm. %Implementation details are described in Appendix \ref{AppB}. 
\section{Methods}
\label{methods}
Here we describe our branch and bound algorithm that searches over a space of $2^p$ models to find minimum and maximum (adjusted) slope coefficients for the effect of $x$ on $y$. The algorithm will be applied to a data set with $n$ rows and $2+p$ columns; see Table \ref{tab1a}. There is a column vector $y$ of the response variable, a column vector $x$ of the explanatory variable of interest, and a matrix $s$ of covariate column vectors $s_1$ through $s_{p}$ each corresponding with a measured covariate or an interaction variable built from measured covariates. We assume throughout that all vectors are linearly independent. Also, without loss of generality we assume that the mean of the entries of any column vector is zero. In practice we pre-process the data to ensure mean-zero columns of data.

\begin{table}[h!]
\centering
\caption{An example of a data set}
\label{tab1a}
\begin{tabular}{rcccc}
\toprule
% & \multicolumn{2}{c}{Diabetes}\\
%\cmidrule{2-3}
$y$&$x$&$s_1$&$\cdots$&$s_p$\\
\hline
$y_1$&$x_1$&$s_{1_1}$&$\cdots$&$s_{p_1}$\\
$y_2$&$x_2$&$s_{1_2}$&$\cdots$&$s_{p_2}$\\
$\vdots$&$\vdots$&$\vdots$&$\ddots$&$\vdots$\\
$y_n$&$x_n$&$s_{1_n}$&$\cdots$&$s_{p_n}$.\\
\hline
\end{tabular}
\end{table}

At each stage of the algorithm we will analyze two disjoint $n\times p_w$ and $n\times p_z$ submatrices $w$ and $z$ of $s$. That pair $(w,z)$ of submatrices is variable, depending on the stage of the algorithm. At the initial stage we set $p_w=0$ so that $w=\emptyset$ while $p_z=p$ and $z=s$. At all stages, because $w$ and $z$ are disjoint, we have $p_w+p_z\leq p$. We keep track of the disjoint submatrices with disjoint subindexes $I_w=\{w_1,\cdots,w_{p_w}\}$ and $I_z=\{z_1,\cdots,z_{p_z}\}$ of the main covariate index $I_s=\{1,\cdots,p\}$. Here $I_w$ indexes the columns of $w$ and $I_z$ indexes the columns of $z$. At any given stage the columns of $w$ are the covariate vectors corresponding with covariates included within the model at that stage, while the columns of $z$ are the covariate vectors corresponding with covariates that are candidates for inclusion within an extension of the model at that stage.%The submatrices $w$ and $z$ are variable, and the pair $(w,z)$ depends on the stage of the algorithm. 

Given $(w,z)$ we define vectors of fitted values, $\hat{y}_w=w(w^tw)^{-1}w^ty$ and $\hat{x}_w=w(w^tw)^{-1}w^tx$, and also the matrix $\hat{z}_w=w(w^tw)^{-1}w^tz$. The columns of $\hat{z}_w$ are vectors of fitted values corresponding with the columns of $z$. Let $I_{\tilde{z}}$ be a subindex of $I_z$, and write $\tilde{z}$ for the corresponding submatrix of $z$ with $n$ rows and $p_{\tilde{z}}$ columns where $p_{\tilde{z}}\leq p_z$. We write $\sigma^2$ for variance, $\rho$ for correlation, $R^2$ for coefficients of determination, and $\beta$ for a slope coefficient of interest, while utilizing subscripts to distinguish between quantities of different objects and models. We compute $R^2$ values by projecting the vector of the second subscript onto the span of the columns of the first subscript. For $\beta$ the vector of the first subscript corresponds with the explanatory variable of interest, the vector of the second subscript corresponds with the response variable, and the columns of the matrix of the third subscript correspond with the covariates. The symbol $\hat{\rho}_{x-\hat{x}_w,y-\hat{y}_w}(\tilde{z}-\hat{\tilde{z}}_w)$ stands for the correlation between variables corresponding with vectors of fitted values after $x-\hat{x}_w$ and $y-\hat{y}_w$ are regressed onto $\tilde{z}-\hat{\tilde{z}}_w$. See \cite{KOA} for additional details.

By applying \cite[Proposition 1.1]{KOA} to residual vectors we obtain
\begin{equation}\label{formula}
\beta_{x-\hat{x}_w,y-\hat{y}_w;\tilde{z}-\hat{\tilde{z}}_w}=\frac{\sigma_{y-\hat{y}_w}}{\sigma_{x-\hat{x}_w}}\left(\frac{\rho_{x-\hat{x}_w,y-\hat{y}_w}-R_{\tilde{z}-\hat{\tilde{z}}_w,x-\hat{x}_w}R_{\tilde{z}-\hat{\tilde{z}}_w,y-\hat{y}_w}\hat{\rho}_{x-\hat{x}_w,y-\hat{y}_w}(\tilde{z}-\hat{\tilde{z}}_w)}{1-R^2_{\tilde{z}-\hat{\tilde{z}}_w,x-\hat{x}_w}}\right).
\end{equation}
Since $\hat{\rho}_{x-\hat{x}_w,y-\hat{y}_w}(\tilde{z})$ is a correlation coefficient, and since $R^2$ values are non-negative and non-decreasing when additional explanatory variables are utilized, we have 
\begin{subequations}
\label{bounds}
\begin{align}
-1&\leq \hat{\rho}_{x-\hat{x}_w,y-\hat{y}_w}(\tilde{z})\leq 1,\\
0&\leq R^2_{\tilde{z}-\hat{\tilde{z}}_w,x-\hat{x}_w}\leq R^2_{z-\hat{z}_w,x-\hat{x}_w},\textrm{~and}\\
0&\leq R^2_{\tilde{z}-\hat{\tilde{z}}_w,y-\hat{y}_w}\leq R^2_{z-\hat{z}_w,y-\hat{y}_w}.
\end{align}
\end{subequations}
The constant-time algorithm of \cite[Proposition 2.2]{KOA} analyzes (\ref{formula}) and takes the bounds of (\ref{bounds}) as inputs to produce bounds $l(w,z)$ and $u(w,z)$ as outputs that satisfy
\begin{equation}\label{lu}
    l(w,z)\leq \beta_{x-\hat{x}_w,y-\hat{y}_w;\tilde{z}-\hat{\tilde{z}}_w}\leq u(w,z).
\end{equation} 
%The bounds of (\ref{lu}) hold for any and every subset $z_s\subseteq z$.

Our branch and bound algorithm searches over a space of models by utilizing a queue and repeatedly updating the variables 
\[\textrm{lower}(\beta_{x,y;\cdot})\textrm{~and~}\textrm{upper}(\beta_{x,y;\cdot}),\]
both initialized at $\beta_{x,y;\emptyset}$. %Each item consists of a model determined by $w$ and a set of potential covariates $z$. 
The items of the queue are $(I_w,I_z)$ pairs, and for a given item we obtain the matrices $w$ and $z$ from the indexes $I_w$ and $I_z$. Subsequent computations are then done on $w$, $z$, $x$, and $y$. 

The first item to be pushed into the queue is $(I_w,I_z)=(I_\emptyset,I_s)$. When an item is popped from the queue the following operations are performed. First, update
\begin{subequations}
\label{recops}
    \begin{align*}
        &\textrm{lower}(\beta_{x,y;\cdot})=\min \{\textrm{lower}(\beta_{x,y;\cdot}),\beta_{x-\hat{x}_w,y-\hat{y}_w}\} \textrm{~and}\\
        &\textrm{upper}(\beta_{x,y;\cdot})=\max \{\textrm{upper}(\beta_{x,y;\cdot}),\beta_{x-\hat{x}_w,y-\hat{y}_w}\}.
    \end{align*}
\end{subequations}
Second, via (\ref{lu}), only if
\begin{equation}
\label{pot}
    l(w,z)<\textrm{lower}(\beta_{x,y;\cdot})\textrm{~or~}u(w,z)>\textrm{upper}(\beta_{x,y;\cdot}),
\end{equation}
determine 
%and if either of the inequalities of (\ref{pot}) hold true we then say that $z$ has \textit{potential} given $w$ and $(\textrm{lower}(\beta_{x,y;\cdot}),\textrm{upper}(\beta_{x,y;\cdot})$. In the absence of potential we discard the current item. If there is potential we then proceed to determine
\begin{equation}
\label{bstar}
z_\star=\argmax_{z_i\in I_z}|\rho_{z_i-\hat{z_i}_w,x-\hat{x}_w}\rho_{z_i-\hat{z_i}_w,y-\hat{y}_w}|
\end{equation}
and place the additional items $(I_w,I_z\setminus z_\star)$ and $(I_w\cup z_\star,I_z\setminus z_\star)$ into the queue. The next item is then popped from the queue. The algorithm continues until the queue is empty, and the final values of the temp variables are returned as lower and upper bounds for $\beta_{x,y;\tilde{s}}$, where $\tilde{s}$ is any $n$-row submatrix of the covariate matrix $s$. Pseudocode is provided in Algorithm \ref{alg:bounds}. Implementations are provided in R and Python at \cite{Hughes24}.

\begin{algorithm}
  \caption{Bounding $\beta$}
  \begin{algorithmic}[1]
    \Statex
    \Function{Bounds}{$x, y, s$}
      \State $\textrm{lower}(\beta_{x,y;\cdot}), \textrm{upper}(\beta_{x,y;\cdot}) \gets \beta_{x,y;0}, \beta_{x,y;0}$
      \State $node \gets (I_\emptyset,I_s)$
      \State $frontier \gets \textrm{a FIFO queue}$
      \State $\textproc{push}(node)$
      \While{$\textbf{not}~\textproc{isEmpty}(frontier)$}
        \State $node \gets \textproc{pop}(frontier)$
        \State $I_w,I_z \gets node$
        \State $w,z \gets s[,I_w], s[,I_z]$
        \State $\textrm{lower}(\beta_{x,y;\cdot}) \gets \min(\{\textrm{lower}(\beta_{x,y;\cdot}),\beta_{x-\hat{x}_w,y-\hat{y}_w}\})$
        \State $\textrm{upper}(\beta_{x,y;\cdot}) \gets \max(\{\textrm{upper}(\beta_{x,y;\cdot}),\beta_{x-\hat{x}_w,y-\hat{y}_w}\})$
        \If{$l(w,z)<\textrm{lower}(\beta_{x,y;\cdot})\textrm{~\textbf{or}~}u(w,z)>\textrm{upper}(\beta_{x,y;\cdot})$}
           \State $z_\star=\argmax_{z_i\in I_z}|\rho_{z_i-\hat{z_i}_w,x-\hat{x}_w}\rho_{z_i-\hat{z_i}_w,y-\hat{y}_w}|$
           \State $\textproc{push}((I_w,I_z\setminus z_\star))$
           \State $\textproc{push}((I_w\cup z_\star,I_z\setminus z_\star))$
        \EndIf
      \EndWhile
      \State \Return{$\textrm{lower}(\beta_{x,y;\cdot}), \textrm{upper}(\beta_{x,y;\cdot})$}
    \EndFunction
  \end{algorithmic}
  \label{alg:bounds}
\end{algorithm}

\section{Results}
\label{results}
In this section we state and prove a technical and mathematical result that supports our branch and bound algorithm. We also describe the results of some trial applications. We denote our branch and bound algorithm with BB, and we compare its performance with that of a brute force search algorithm (denoted with BF) that fits every possible model. We assess performance by measuring run times (see Table \ref{tab2}) and numbers of items that pass through the queue (see Table \ref{tab3}) for each of four example applications of the algorithm. The example data sets are described in Appendix \ref{AppA}. The trials were conducted using R version 4.3.1 on a 2017 MacBook Pro with a 2.9 GHz Quad-Core Intel Core i7 processor, 16 GB 2133 MHz LPDDR3 of memory, and 300 GB of disk space. 
\begin{proposition}
\label{Danny}
With the definitions of Section \ref{methods} we have the following identities:
\begin{subequations}
\begin{align}
\beta_{x,y;w,z}&=\beta_{x-\hat{x}_w,y;z-\hat{z}_w}\label{one}\\
&=\beta_{x-\hat{x}_w,y-\hat{y}_w;z-\hat{z}_w}\label{two}.
\end{align}
\end{subequations}
\end{proposition}
\begin{proof}
Without loss of generality we may assume that all vectors have been centered, i.e. we assume that all vectors are mean-zero vectors. %In what follows, least-squares projections of vectors are described with standard, hat notation of the form $\hat{z}_{z_1}$, where $z$ is the response vector and the subscript $z_1$ specifies the explanatory vectors. 
We think of the hat $\hat{~}$ as a projection operator. When a hat appears above a matrix that means the projection is applied to each of the columns of that matrix. The projection is onto a space indicated with a subscript. In what follows $w$ and $z$ are matrices, and $\beta_{w,y;x,z}$ and $\beta_{z,y;x,w}$ are vectors of fitted slope coefficients; see Section \ref{methods}. There exists another vector of slope coefficients $\beta_{w,y;(x-\hat{x}_w),(z-\hat{z}_w)}$ satisfying (\ref{next}) below: 
\begin{subequations}
\label{all}
\begin{align}
\label{first}
\hat{y}_{x,w,z}&:=\beta_{x,y;w,z}x+ w\beta_{w,y;x,z}+z\beta_{z,y;x,w}\\
&=\beta_{x,y;w,z}(\hat{x}_w+(x-\hat{x}_w))+w\beta_{w,y;x,z}+(\hat{z}_w+z-\hat{z}_w)\beta_{z,y;x,w}\\
&=\beta_{x,y;w,z}(x-\hat{x}_w)+(\beta_{x,y;w,z}\hat{x}_w+w\beta_{w,y;x,z}+\hat{z}_w\beta_{z,y;x,w})+(z-\hat{z}_w)\beta_{w,y;x,z}\\
&\label{next}=\beta_{x,y;w,z}(x-\hat{x}_w)+w\beta_{w,y;(x-\hat{x}_w),(z-\hat{z}_w)}+(z-\hat{z}_w)\beta_{w,y;x,z}\\
\label{last}&=:\hat{y}_{(x-\hat{x}_w),w,(z-\hat{z}_w)}=\hat{y}_{(x-\hat{x}_w),(z-\hat{z}_w)}+\hat{y}_w.
\end{align}
\end{subequations}
The right equality in (\ref{last}) holds since the columns of $w$ are orthogonal to both $x-\hat{x}_w$ and the columns of $z-\hat{z}_w$. The coefficient of $x$ in (\ref{first}) is the same as the coefficient of $(x-\hat{x}_w)$ in (\ref{next}), demonstrating the truth of (\ref{one}). Moreover, replacing $y$ with $y-\hat{y}_w$ affects only the coefficients of $\beta_{w,y;(x-\hat{x}_w),(z-\hat{z}_w)}$, demonstrating the truth of (\ref{two}).
\end{proof}
\noindent The identity in (\ref{two}) ensures that adjustment for the covariates of $w$ and $z$ can be accomplished by computing vectors of residuals after regression onto $w$ and then conducting a regression with those residual vectors and adjusting for only the residual covariates of $z$, thus justifying the formulations in (\ref{formula}) and (\ref{lu}). %Proposition \ref{Danny} is proven in Appendix \ref{AppA}.
%
%The tables below summarize the results of trials of BFS, BB, and BBR algorithms, and then BBR+QR, BBR+Indexing, BBR+Separate, and finally the optimal BBR+QR+Indexing+Separate, on various data sets. Section \ref{applications} and Appendix \ref{AppA} provide further descriptions of the data sets and how they were obtained. The trials were conducted using R version 4.3.1 on a 2017 MacBook Pro with a 2.9 GHz Quad-Core Intel Core i7 processor, 16 GB 2133 MHz LPDDR3 of memory, and 300 GB of disk space. 
%
%In Table \ref{tab2} we see that with the SUPPORT2 data set and its thirty covariates, the only algorithm to finish was the BBR algorithm. Both branch and bound algorithms spend more time at each node, but, as shown in Table \ref{tab3}, when the BBR algorithm is applied the conditions of (\ref{gencond}) frequently hold and the number of visited nodes is drastically reduced. We describe additional ways to improve the run times in Section \ref{discussioneff}.

\begin{table}[h!]
\centering
\caption{Run times in seconds (s) and minutes (min) of trials of the branch and bound (BB) algorithm and the brute force (BF) search algorithm, each applied to various data sets with \(n\) observations and \(p\) covariates}
\label{tab2}
\begin{tabular}{rcccc}
\toprule
 &&& \multicolumn{2}{c}{Algorithm}\\
\cmidrule{4-5}
Data Set&$n$&$p$&BF&BB\\
\hline
NHANES\textsubscript{1} & 14,208 & 10 & 11.281 s & 37.788 s\\
AllCountries & 215 & 21 & 50.493 min & 5.025 min\\
NHANES\textsubscript{2} & 4269 & 27 & $>180$ min & 62.834 min\\
SUPPORT2 & 9105 & 30 & $>180$ min & 42.354 min\\
\hline
\end{tabular}
\end{table}

\begin{table}[h!]
\centering
\caption{Number of items that pass through the queue, for trials of the branch and bound (BB) algorithm and the brute force (BF) search algorithm, each applied to various data sets with \(n\) observations and \(p\) covariates}
\label{tab3}
\begin{tabular}{rcccc}
\toprule
 &&& \multicolumn{2}{c}{Algorithm}\\
\cmidrule{4-5}
Data Set&$n$&$p$&BF&BB\\
\hline
NHANES\textsubscript{1} & 14,208 & 10 & 1024 & 1323\\
AllCountries & 215 & 21 & 2,097,152 & 30,537\\
NHANES\textsubscript{2} & 4269 & 27 & 134,217,728 & 78,561\\
SUPPORT2 & 9105 & 30 & 1,073,741,824 & 41,537\\
\hline
\end{tabular}
\end{table}

\section{An example application}
\label{applications}

% Every two years, the National Center for Health Statistics (NCHS) collects nutrition and health data via a complex, multistage, probability sampling design from a representative sample of the civilian, non-institutionalized US population \cite{NCHS}. This data collection process is called the National Health and Nutrition Examination Survey (NHANES). %, and it is widely used to gain insight into possible causes and effects of health. 
% Each participant is interviewed twice: once at home and once in a medical examination facility. We have utilized data from NHANES twice, first in what is labeled as NHANES\textsubscript{1} and second in what is labeled as NHANES\textsubscript{2}. Here we are focused exclusively on describing the details of our example application to the NHANES\textsubscript{2} data set. This more detailed example application is meant to demonstrate our methodology. Our results are not meant to support any definitive scientific conclusions.

Every two years, the Centers for Disease and Control Prevention (CDC) through the National Center for Health Statistics (NCHS) collects nutrition and health data via a complex, multistage, probability sampling design from a representative sample of the civilian, non-institutionalized US population \cite{NCHSguidelines2010}. This data collection process is called the National Health and Nutrition Examination Survey (NHANES). Each participant is interviewed twice: once at home and once in a medical examination facility. We have utilized data from NHANES twice, first in what is labeled as NHANES\textsubscript{1} and second in what is labeled as NHANES\textsubscript{2}. Here we are focused exclusively on describing the details of our example application to the NHANES\textsubscript{2} data set. This more detailed example application is meant to demonstrate our methodology. Our results are not meant to support any definitive scientific conclusions.

Within this detailed example application we have blood serum vitamin D (SD) levels as our $x$ variable and body mass index (BMI) as our $y$ variable. The SD levels were measured in nmol/L and assessed using the 25-hydroxy vitamin D test, a measure of both vitamin D2 and vitamin D3 in the blood \cite{Cavalier2021}. BMI was defined as a participant's weight in kilograms divided by the square of their height in meters ($\textrm{kg/m}^2$). Previous research indicates that a deficiency of SD is associated with obesity \cite{Vanlint2013}. Figure \ref{fig2} illustrates that association with a scatterplot showing negative correlation between SD concentration and BMI. While the slope coefficient is negative in the case of simple regression, it is possible that adding a set of covariates to the model may cause the association to reverse as described in \cite{Patel2015} and \cite{KD}. %The scientific literature on this topic is ambiguous \cite{Karampela2021,Vrani2019}.

\begin{figure}[h]
    \centering
    \includegraphics[width=.7\linewidth]{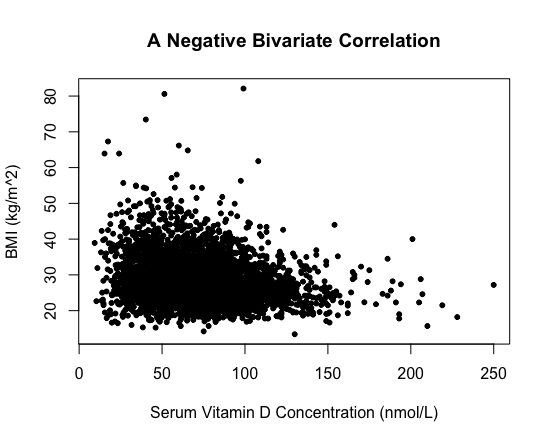}
    \caption{A scatterplot of SD and BMI.}
    \label{fig2}
\end{figure}

It is difficult to arrive at a definitive scientific conclusion when model uncertainty exists. We consider the following measured covariates: fish consumption, vitamin D supplementation, exercise activity level, age, gender, and two racial indicator variables (one for black and the other for white). Perhaps (outdoor) physical activity raises SD, but only amongst individuals with light skin, while lowering BMI. Likewise, perhaps some dietary variables increase SD while raising or lowering BMI, but only for younger individuals. For these hypothetical reasons, we include not just the seven covariate vectors of raw data, but many of their interaction terms as well. From these, we exclude only the interaction between the two racial indicators. The result is \[\left(\binom 7 2 - 1\right) + 7 = 27\] total covariate vectors.

The simplest model where we regress BMI onto SD yields a slope coefficient of $\beta_{\textrm{SD,BMI}}=-0.0496$ $\tfrac{\textrm{kg/m}^2}{\textrm{nmol/L}}$. However, we may adjust that slope coefficient by conditioning on any of $2^{27}=1,342,17,728$ different subsets of covariates. Instead of checking every possible adjustment, we apply our branch and bound algorithm to conclude that
\begin{equation}\label{tightbounds}-0.0581 \,\tfrac{\textrm{kg/m}^2}{\textrm{nmol/L}}\leq \beta_{\textrm{SD;BMI};\cdot}\leq -0.0440 \, \tfrac{\textrm{kg/m}^2}{\textrm{nmol/L}}.\end{equation}
The bounds on $\beta_{\textrm{SD,BMI};\cdot}$ are surprisingly tight, illustrating the potential of our branch and bound algorithm to inform researchers. 

% \textcolor{red}{Here is my prefered outline for this Application section: paragraph one should introduce the details of NHANES. Paragraph two should introduce x and y and reference a picture of a scatterplot. Paragraph 3 should emphasize the potential for reversals of least squares estimates, model uncertainty, and ambiguous scientific literature on this topic of whether x causes y. Paragraph 4 should describe how to select covariates. Paragraph 5 should talk about interaction (place your additional info there, compressed, if needed, and put extra stuff to the appendix). Paragraph 6 should describe how our algorithm can shine some light on the instability or uncertainty of the situation, and in paragraph six we can describe the results of the algorithm. If you need to save space and compress things then merge paragraphs 4 and 5 above. Again, anything extra can go to the appendix.}

\section{Discussion}
\label{discussion}
We have introduced a branch and bound algorithm to support interpretation of regression coefficients in the presence of model uncertainty. The branch and bound algorithm is based on the confounding interval of \cite{KOA} as described in Section \ref{methods}. It is supported by Proposition \ref{Danny} of Section \ref{results}. The results of trials described in Table \ref{tab2} of Section \ref{results} show that the branch and bound algorithm can be applied in situations where a brute force search is not applicable. The example application of Section \ref{applications} shows the practical utility of the branch and bound algorithm. Before application of the algorithm a reversal of the correlation from negative to positive was deemed plausible and possibly consistent with the measured covariate data, but after application of the algorithm the tight bounds of (\ref{tightbounds}) provide some assurance that the negative correlation is indeed stable and consistent with many models that could be fit to the observed data.

This paper was primarily concerned with demonstrating a proof of concept: we can bound slope coefficients over a large space of models utilizing a branch and bound methodology. When updating temporary maximum or minimum values for slope coefficients, however, we did not conduct any regression diagnostics. We did center vectors to make their means zero, but we did not transform any variables, e.g. by utilizing logarithmic or other transformations. We did not check to see if models were well fit, nor did we check for outliers. Our code can be modified to include such diagnostic checks and transformations if so desired. Our algorithm is meant to be applied to census data or data from a large representative sample. The methodology can be generalized for application to small sample data but that is beyond the scope of this work. 

Our branch and bound algorithm outperforms a brute force search primarily due to its bound that is applicable whenever $z$ lacks potential. We may say that $z$ has potential given $w$, $\textrm{lower}(\beta_{x,y;\cdot})$, and $\textrm{upper}(\beta_{x,y;\cdot})$ when either condition in (\ref{pot}) holds. Table \ref{tab3} shows some evidence that the bound is often applicable, i.e. that $z$ often lacks potential. We select $z_\star$ in (\ref{bstar}) with a product of correlations supported by \cite{EJASA}, and that prioritized selection is meant to increase the frequency with which $z$ lacks potential which reduces the run time. QR decompositions may be utilized to further enhance our algorithm, and if there is sufficient storage space then additional efficiency gains may be possible by storing vectors within the fields of queued items. 

During each of our four trials we found situations where $z$ lacks potential, but in general there is no guarantee. The branch and bound algorithm can be used whenever the principle of least squares is appropriate. We can apply our algorithm with models that are linear in their parameters. We can not directly apply our algorithm with logistic regression models, but our algorithm can be applied with polynomial regression models. It is straightforward to include higher-order interaction terms between covariates as additional vectors of covariate data. Inclusion of interaction terms that involve $x$ however is more complicated, as interpretation of the regression coefficient $\beta_{x,y;\cdot}$ would be affected, although such an approach could help distinguish between direct and indirect effects.

The confounding interval of \cite{KOA} was designed to handle unmeasured confounding. In our example application to assess the stability of the apparently protective effect of vitamin D against obesity, the bounds of (\ref{tightbounds}) were computed over subsets of measured covariates, including dietary variables, a physical activity variable, and four background characteristics. A skeptic may be concerned because important genomic variables were left out of the analysis. There are techniques for handling all possible, unmeasured confounders \cite{Knaeble2023,arxivOCT,arxivPartial}. In a regression context, \cite{Oster} has utilized an $R_{\textrm{max}}$ parameter to bound $R^2$ parameters strictly below one, which can be justified even while accounting for unmeasured covariates (see \cite[Theorem 3.3]{arxivEntropy} and \cite[Proposition 3.1]{arxivPartial}). However, we are unaware of a methodology that combines those techniques with a search over subsets of measured covariates and doesn't reduce to the basic confounding interval methodology of \cite{KOA}. We suspect that reduction to occur because conditioning on $w=\emptyset$ provides the most freedom within the feasible unmeasured covariate vector space to construct the most extreme confounding interval; see the mathematical details of \cite{KOA}. We think of our introduced branch and bound algorithm less as a methodology to handle unmeasured confounding and more as a technique to summarize an existing data set in the presence of model uncertainty.

\appendix

\section{Data Sets}
\label{AppA}
We applied our algorithm to four data sets which are described in the sub sections below. Readers desiring to eliminate the nonresponse and oversampling biases characteristic of the NHANES survey data may follow the guidelines given by the CDC in \cite{NCHSguidelines2010} and \cite{NCHSguidelines2012}.

\subsection{NHANES\textsubscript{1} Data Set}
The first data set we made use of was the NHANES\textsubscript{1} data set which featured $14,208$ observations from the years 1999 to 2004. The data was wrangled and cleaned by Patel et al. \cite{Patel2015} and posted on GitHub. For our algorithm, we considered BMI (\(y\)) regressed onto total calories consumed the day prior, as a proxy for daily caloric intake (\(x\)). We included 10 covariates: serum calcium levels (mg/dL), serum magnesium levels (mg/dL), serum potassium levels (mg/dL), serum phosphorus levels (mg/dL), serum sodium levels (mg/dL), fiber consumed the day prior (gm), iron consumed the day prior (mg), cholesterol consumed the day prior (mg), zinc consumed the day prior (mg), and thiamin (vitamin B1) consumed the day prior (mg).

\subsection{AllCountries Data Set}
The second data set we used was the AllCountries data set provided in the Lock5Data package in R \cite{Lock2017}. This data set featured 215 observations each representing individual countries. Data was collected for 2018 or the most recently available year. For our algorithm, we considered the percentage of government expenditures directed towards healthcare (\(y\)) regressed onto the percentage of the population at least 65 years old (\(x\)). We considered 21 covariates: country size (1000 square kilometers), country population (millions), density (no. people per square kilometer), gross domestic product per capita (USD), percentage of population living in rural areas, CO2 emissions (metric tons per capita), price for a liter of gasoline (USD), percentage of government expenditures directed towards the military, number of active duty military personnel (in 1000's), percentage of the population with access to the internet, cell phone subscriptions (per 100 people), percentage of the population with HIV, percentage of the population considered undernourished, percentage of the population diagnosed with diabetes, births per 1000 people, deaths per 1000 people, average life expectancy (years), percent of females 15-64 in the labor force, percent of labor force unemployed, kilotons of oil equivalent, and electric power consumption (kWh per capita).

\subsection{SUPPORT2 Data Set}
The third data set we utilized was the SUPPORT2 data set from the Vanderbilt University Department of Biostatistics \cite{Harrell2022}. This data set features $9,105$ observations representing critically ill hospital patients with advanced stages of severe illnesses from 5 medical centers. Data was collected from 1989 to 1991 and from 1992 to 1994. For our algorithm, we considered the number of days from entry into the study to discharge from the hospital (\(y\)) regressed onto white blood count (in thousands) measured at day 3 (\(x\)). We considered a total of 30 covariates: age, sex, race (white), race (black), days of follow-up, the number of simultaneous diseases exhibited by the patient, whether the patient's income exceeds \$50k per year, SUPPORT day 3 coma score based on Glasgow scale, hospital charges, total ratio of costs to charges, average TISS score from days 3-25, SUPPORT physiology score on day 3, APACHE III day 3 physiology score, SUPPORT model 2-month survival estimate on day 3, SUPPORT model 6-month survival estimate on day 3, whether the patient has diabetes, whether the patient has dementia, whether the patient has cancer, physician's 2-month survival estimate for the patient, physician's 6-month survival estimate for the patient, whether the patient has a do not resuscitate (DNR) order, day of DNR order, mean arterial blood pressure of the patient on day 3, heart rate of the patient on day 3, respiration rate of the patient on day 3, temperature (C) on day 3, serum creatinine on day 3, serum sodium concentration on day 3, arterial blood pH, and imputed ADL calibrated to surrogate.

\subsection{NHANES\textsubscript{2}}
The fourth data set utilized was the NHANES\textsubscript{2} data set which featured $4,269$ observations from the years 2007 to 2012. It is described in part in Section \ref{applications} and is available at the second author's GitHub \cite{Hughes24}. As mentioned in Section \ref{applications}, we considered body mass index (BMI) (our \(y\) variable) regressed onto serum vitamin D concentration (our $x$ variable). The names of the measured covariates are described in Section \ref{applications}. Here we describe some additional details. 23 variables indicating the frequency at which 23 fish species were eaten over the 30 days prior to the interview were summed together to produce a single fish score covariate. To measure vitamin D supplementation, we averaged participants' dosage of vitamin D supplements for the day before the at-home and on-site interviews. If a response was recorded for one of these days but not the other, we took the dosage for the recorded day. To gauge participants' physical activity, the NCHS suggests combining via a weighted sum 5 variables measuring minutes of moderate and vigorous exercise during work and recreation and minutes spent walking or biking during commuting in a typical week (see Appendix 1 of \cite{NCHS2012b}). We used the provided weights to combine these variables into a single measure of equivalent task (MET) score. After all the data was transformed, observations with missing data were dropped from the data set, leaving the $4,269$ observations mentioned above.
\section*{Acknowledgement}
Thank you to Kaleb Allen, Jonah Cragun, Nate Lovett, and Ammon Price for helping to implement our branch and bound algorithm.
\bibliographystyle{plain}

\end{document}